\documentclass[final,1p,times]{elsarticle}
\usepackage{geometry}
\usepackage{amsthm}
\usepackage{colortbl}
\usepackage{xcolor}

\usepackage{algpseudocode}
\usepackage{algorithm}
\usepackage{bm}
\usepackage{multirow}
\usepackage{makecell}
\usepackage{soul}
\usepackage{enumitem}
\usepackage{threeparttable,multirow,dcolumn,booktabs}

\usepackage{amssymb}
\usepackage{pifont}
\journal{J. Comp. Phys.}
\usepackage{amsfonts,amssymb,amsbsy,amsmath,tabulary,graphicx,times,caption,fancyhdr}
\usepackage[utf8]{inputenc}
\usepackage{url,morefloats,floatflt,cancel,tfrupee}
\usepackage{mathrsfs}
\usepackage{mathtools,relsize}
\usepackage{graphicx}
\newtheorem{theorem}{Theorem}[section]
\newtheorem{lemma}[theorem]{Lemma}

\usepackage{hyperref}
\usepackage{setspace}
\usepackage{color}

\begin{document}

\begin{frontmatter}

\title{Random Batch Sum-of-Gaussians Method for Molecular Dynamics of Born-Mayer-Huggins Systems}

\author{Chen Chen}\ead{ronnie7@sjtu.edu.cn} 
\author{Jiuyang Liang} 
\author{Zhenli Xu} 
\author{Qianru Zhang} 

\begin{abstract}
The Born-Mayer-Huggins (BMH) potential, which combines Coulomb interactions with dispersion and short-range exponential repulsion, is widely used for ionic materials such as molten salts. However, large-scale molecular dynamics simulations of BMH systems are often limited by computation, communication, and memory costs. We recently proposed the random batch sum-of-Gaussians (RBSOG) method, which accelerates Coulomb calculations by using a sum-of-Gaussians (SOG) decomposition to split the potential into short- and long-range parts and by applying importance sampling in Fourier space for the long-range part. In this work, we extend the RBSOG to BMH systems and incorporate a random batch list (RBL) scheme to further accelerate the short-range part, yielding a unified framework for efficient simulations with the BMH potential. The combination of the SOG decomposition and the RBL enables an efficient and scalable treatment of both long- and short-range interactions in BMH system, particularly the RBL well handles the medium-range exponential repulsion and  dispersion by the random batch neighbor list. Error estimate is provided to show the theoretical convergence of the RBL force. We evaluate the framework on molten NaCl and mixed alkali halide with up to $5\times10^6$ atoms on $2048$ CPU cores. Compared to the Ewald-based particle-particle particle-mesh method and the RBSOG-only method, our method achieves approximately $4\sim10\times$ and $2\times$ speedups while using $1000$ cores, respectively, under the same level of structural and thermodynamic accuracy and with a reduced memory usage. These results demonstrate the attractive performance of our method in accuracy and scalability for MD simulations with long-range interactions.
\end{abstract}
\end{frontmatter}

\section{Introduction}
Molten alkali halides, inorganic liquids composed of alkali-metal cations and halide anions, are important electrolytes and heat-transfer fluids in advanced energy systems due to their high thermal stability, strong ionic conductivity, and chemical durability. Since their early use as reactor coolants at Oak Ridge National Laboratory in the 1960s~\cite{haubenreich1964msre,haubenreich1970experience}, molten salts have remained central to Generation-IV reactor designs, motivating extensive computational studies to understand the structure and transport properties of molten salts~\cite{defever2020melting,siemer2015molten}. Due to experimental challenges such as corrosion, extreme temperatures, and high costs, molecular dynamics (MD) simulation plays a central role in predicting the atomic-scale structure, transport, and thermodynamic properties of these systems~\cite{wang2014molecular,wang2015molecular,yang2023new,song2022molecular}.

Molten alkali halides are dense ionic liquids in which long-range electrostatics, short-range Pauli repulsion, and dispersion all matter. To reproduce the experimental structure and transport, dedicated force fields are required. The Born-Mayer-Huggins (BMH) pairwise potential~\cite{born1932gittertheorie,huggins1937lattice} is often employed for the task:
\begin{equation}\label{eq:BMH}
	U_{\text{BMH}}(r_{ij})=\frac{q_iq_j}{r_{ij}}+A_{ij}\, \exp\left(\dfrac{\sigma_{ij}-r_{ij}}{\rho}\right)-\dfrac{C_{ij}}{r_{ij}^6}-\dfrac{D_{ij}}{r_{ij}^8}.
\end{equation}
Here, $r_{ij}$ denotes the distance between ions $i$ and $j$; $q_i$ is the ionic charge; the exponential term models short-range Pauli repulsion with amplitude $A_{ij}$, range parameter $\rho$, and $\sigma_{ij}:=\sigma_i+\sigma_j$ being the addition of ionic radii; and the last two terms describe dipole-dipole and dipole-quadrupole dispersion with coefficients $C_{ij}$ and $D_{ij}$. Despite their broad use in molten-salt MD simulations~\cite{porter2022computational,galashev2023molecular,shore2000simulations}, large-scale BMH simulations remain computationally costly. Several fundamental bottlenecks hinder large-scale BMH simulations. The long-range Coulomb term, typically handled by communication-intensive Ewald-type solvers with fast Fourier transforms (FFTs)~\cite{hockney2021computer,darden1993particle}, introduces significant scalability limitations~\cite{ayala2021scalability}. Meanwhile, the exponential repulsion and dispersion terms exhibit slow algebraic decay, representing a medium-range interaction that requires a large real-space cutoff and converges slowly in Fourier space. Moreover, periodic fast multipole methods (FMMs)~\cite{yan2018flexibly,pei2023fast} further introduce complexity by splitting the periodic system into near‐field and far‐field regions; the near‐field interactions must be evaluated over the $27$ neighboring image cells, which is itself computationally costly. Recently proposed dual-space multilevel kernel-splitting (DMK) framework~\cite{jiang2025dual} and Ewald summation with prolates (ESP) method~\cite{liang2025accelerating} provide promising alternatives, but have not yet been adapted to BMH systems. These factors collectively limit both the accessible time scale and system size in BMH-based MD simulations.

In recent years, stochastic algorithms have emerged and improved scalability by replacing global kernel sums with unbiased mini-batch estimators, and their accuracy has been rigorously established for interacting particle systems~\cite{jin2020random,jin2023ergodicity,ye2024error,cai2024convergence,lin2024hybrid}. For electrostatics, random batch Ewald (RBE)~\cite{jin2021random,liang2022superscalability,liang2022random} avoids global FFTs via importance sampling while keeping linear scaling. Since the construction of an SOG approximation has been extensively studied~\cite{greengard2018anisotropic,liang2020kernel,beylkin2005approximation}, the random batch sum-of-Gaussians (RBSOG) method~\cite{liang2023random,chen2025random} has been proposed to replace the Ewald split with a sum-of-Gaussians (SOG) decomposition of $1/r$, removing the discontinuity of the Ewald split and further reducing batch size. However, the other terms of the BMH potential (exponential repulsion and dispersion) can remain the dominant cost because of their medium-range nature, creating a bottleneck that is not addressed by long-range solvers alone. The random batch list (RBL) scheme, developed to accelerate the short-range potentials such as the Lennard‑Jones (LJ) potential~\cite{liang2021random} and the embedded atom potential~\cite{zhang2025random}, reduces the neighbor‑list cost by randomly sampling particles within a cutoff radius. The improved random batch Ewald (IRBE) approach~\cite{liangIRBE} combines RBE with the LJ-based RBL. Notably, the RBMD package~\cite{gao2026rbmd} implements the IRBE on GPU-CPU heterogeneous architectures, enabling efficient all-atom simulations of up to ten million particles on a single GPU. However, IRBE has not been applied to BMH systems and does not take advantage of the more efficient SOG decomposition. A linear-scaling method that treats all components of the BMH potential in a stochastic framework is of great significance for practical simulations. 

In order to overcome the challenge of handling medium-range interactions in BMH systems, we propose a unified framework, the improved random batch sum-of-Gaussians (IRBSOG) method, which seamlessly couples the RBSOG and RBL schemes. The key contributions of this work are threefold. First, we develop an integrated approach for the BMH potential in which the Coulomb kernel is split via a sum-of-Gaussians (SOG) approximation, with the long-range part efficiently evaluated via RBSOG in Fourier space using random batch importance sampling, while the short-range residual, together with the exponential repulsion and dispersion terms, is handled by a core-shell RBL scheme in real space. This synergistic combination specifically targets the multi-scale nature of the BMH potential, enabling scalable and accurate treatment of both interaction regimes. Second, recognizing that short-range interactions in BMH systems exhibit slower decay and require larger cutoffs compared to Lennard-Jones-type potentials, we systematically investigate and optimize the RBL implementation under these conditions. An adaptive, species-dependent cutoff strategy is introduced to account for size differences between alkali-metal cations and halogen anions, effectively reducing neighbor counts and memory usage. Third, we performed systematic MD simulations across a range of BMH systems--including molten NaCl and mixed alkali halides--to thoroughly validate the accuracy, scalability, and transferability of our approach. The proposed IRBSOG framework achieves $O(N)$ complexity with significantly reduced communication and memory overhead compared to both the particle-particle particle-mesh (PPPM) and the RBSOG-only methods. In large-scale benchmarks, it reproduces key structural, dynamical, and thermodynamic properties while delivering up to an order-of-magnitude speedup over PPPM and a roughly twofold improvement over RBSOG-only simulations for systems of up to $5\times10^6$ atoms on $2048$ CPU cores.


The remainder of this paper is organized as follows. In Section~\ref{sec:II}, we present the theoretical foundation of the SOG decomposition and describe the new framework for the BMH potential. Section~\ref{sec:III} validates the accuracy and performance through systematic simulations of molten NaCl and alkali binary halide systems. Conclusion remarks are provided in Section~\ref{sec:IV}.

\section{Method}\label{sec:II}
In this section, we introduce a unified framework for accelerating MD simulations with the BMH potential. Section~\ref{sec:II.A} proposes a new decomposition of the BMH kernel using the SOG approximation of Coulomb kernel. Section~\ref{sec:II.C} details the framework and its implementation. Section~\ref{sec:II.D} presents complexity and convergence analysis.

\subsection{SOG decomposition of the BMH potential}\label{sec:II.A}
Consider a system of $N$ charged particles located at positions $\{\boldsymbol{r}_i = (x_i,y_i, z_i)\}_{i=1}^{N}$ with charges $\{q_i\}_{i=1}^{N}$ in a cuboid box of side length $L_x$, $L_y$, and $L_z$. We assume the charge neutrality condition, $\sum_{i=1}^{N}q_i=0$. Periodic boundary conditions are imposed to mimic bulk environments. The total energy of the system is given by
\begin{equation}
U=\frac{1}{2}\sum_{i,j=1}^N\sum_{\bm{n}\in\mathbb{Z}^3}\,^{\prime}U_{\text{BMH}}\left(|\bm{r}_{ij}+\bm{n}\circ \bm{L}|\right),
\end{equation}
where $U_{\mathrm{BMH}}$ is given by Eq.~\eqref{eq:BMH}, $\bm{r}_{ij}=\bm{r}_i-\bm{r}_j$, $\bm{L}=(L_x,L_y,L_z)$, the prime indicates that $\bm{n}=\bm{0}$ is excluded when $i=j$, and ``$\circ$'' denotes elementwise multiplication. The exponential repulsion and dispersion terms in $U_{\mathrm{BMH}}$ decay rapidly and are evaluated in real space with a cutoff $r_c$. In contrast, the Coulomb contribution is only conditionally convergent, so direct truncation is not valid;  in mainstream MD packages it is typically treated by FFT-based mesh-Ewald methods~\cite{hockney2021computer,darden1993particle}.

Here, we treat the Coulomb part of the BMH potential with an SOG decomposition instead of the classical Ewald decomposition, following the idea of the RBSOG \cite{liang2023random}. We start from the integral representation of a power-law kernel:
\begin{equation}
\dfrac{1}{r^{2\alpha}} = \frac{1}{\Gamma(\alpha)}\int^{\infty}_{-\infty} e^{-e^{\tau}r^2 + \alpha\tau}d\tau,
\end{equation}
where $\alpha>0$ and $\Gamma(\cdot)$ denotes the Gamma function. Setting $\alpha=1/2$ gives the representation of the $1/r$ kernel. Introducing the change of variables $\tau=\ln(x^2/2s^2)$ and using geometrically distributed quadrature nodes $x=b^{-m}$ yields the bilateral SOG approximation (BSA):
\begin{equation}\label{eq::1-r}
\frac{1}{r} \approx \frac{2\ln b}{\sqrt{2\pi s^2}}\sum^{\infty}_{m=-\infty}b^{-m}e^{-\frac{r^2}{2b^{2m} s^2}},
\end{equation}
where $b>1$ is a constant and $s$ controls the bandwidth of Gaussians. A key property of Eq.~\eqref{eq::1-r} is a uniform pointwise relative error bound ($b\to 1$ for all $r>0$)~\cite{beylkin2010approximation}:
\begin{equation}\label{eq::asymperror}
    \left\vert 1-\frac{2r\ln b}{\sqrt{2\pi s^2}}\sum_{m=-\infty}^{\infty}b^{-m}e^{-\frac{r^2}{2b^{2m} s^2}}\right\vert\lesssim2\sqrt{2}e^{-\frac{\pi^2}{2\ln b}}. 
\end{equation}      
Other kernel-independent SOG methods~\cite{greengard2018anisotropic,gao2022kernel} can also be used to construct SOG approximations of $1/r$. By using the BSA, the so-called u-series method~\cite{predescu2020u} splits $1/r$ into short- and long-range parts,
\begin{equation}\label{eq::1r}
\frac{1}{r}\rightarrow\mathcal{N}_b^{s}(r)+\mathcal{F}_b^{s}(r),
\end{equation}
where the long-range part
\begin{equation}\label{eq:far}
    \mathcal{F}_b^{s}(r)=\sum_{m=0}^Mw_{m}e^{-\frac{r^2}{2b^{2m} s^2}}
\end{equation}
is the BSA series truncated to $0\leq m\leq M$ with $w_{m}=\sqrt{2}\ln b/(\sqrt{\pi} b^{m}s)$, and
the short-range component is
\begin{equation}\label{eq::short}
\mathcal{N}_{b}^{s}(r)=\begin{cases}
		\dfrac{1}{r}-\mathcal{F}_{b}^{s}(r),\quad &\text{if }r< r_c,\\[2.2em]0,&\text{if }r\geq r_c,
	\end{cases}
\end{equation}
with cutoff $r_c$ chosen as the smallest root of $r\,\mathcal{F}_b^{s}(r)=1$ to ensure continuity of the potential. For MD, one typically enforces $C^1$ continuity (continuous force) via
\begin{equation}
	\dfrac{1}{r_c^2}-\left.\partial_r\mathcal{F}_b^{s}(r)\right|_{r=r_c}=0.
\end{equation} 
For fixed $b$ and $s$, this can be satisfied by adjusting the coefficient of the narrowest Gaussian
\begin{equation}
	w_0=\dfrac{1}{e^{-{r_c^2}/{(2b^{2m} s^2})}}\left[\dfrac{1}{r_c}-\mathcal{F}_b^{s}(r_c)\right],
\end{equation}
and then solving the potential-continuity condition to determine $r_c$. This redefinition of the narrowest Gaussian weight is necessary to avoid large local errors. Thanks to the increased smoothness of the resulting splitting, the SOG decomposition can achieve accuracy comparable to Ewald-based methods with reduced computational cost, while maintaining long-term energy stability in NVE simulations~\cite{chen2025random,predescu2020u}. Moreover, Gaussian functions are separable and can be expressed as products of one-dimensional functions, a property that can be exploited to reduce the number of sequential communication stages in FFT-based acceleration by roughly a factor of two~\cite{predescu2020u}. These features have motivated the adoption of SOG-based decompositions in a range of applications, including quantum mechanics/molecular mechanics calculations~\cite{laino2005efficient}, machine-learning force fields~\cite{ji2025machine}, and Schr\"{o}dinger equation solvers~\cite{zhou2025sum}.

In this work, we build on the idea of u-series to split the BMH energy into short- and long-range parts: 
\begin{equation}
U=U_{\mathcal{N}}+U_{\mathcal{F}},
\end{equation}
where the long-range part 
\begin{equation}\label{eq:UfBMH}
U_{\mathcal{F}}=\frac{1}{2} \sum_{\boldsymbol{n}}\,^{\prime} \sum_{i, j} q_i q_j \mathcal{F}_{b}^{s}\left(\left|\boldsymbol{r}_{i j}+\boldsymbol{n} \circ \boldsymbol{L}\right|\right)
\end{equation}
is a sum of SOG series, and the remaining short-range part is given by
\begin{equation}
    U_{\mathcal{N}}= \frac{1}{2} \sum_{\boldsymbol{n}}\,^{\prime} \sum_{i, j} \biggl[q_i q_j\, 
        \mathcal{N}_{b}^{s}\bigl(\bigl|\boldsymbol{r}_{i j} + \boldsymbol{n} \circ \boldsymbol{L}\bigr|\bigr)+ \mathcal{B}_{\text{BMH}}\bigl(\bigl|\boldsymbol{r}_{i j} + \boldsymbol{n} \circ \boldsymbol{L}\bigr|\bigr) \biggr].
\end{equation}
The function $\mathcal{B}_{\mathrm{BMH}}(r)$ contains the exponential repulsion and dispersion terms,
\begin{equation}\label{eq::short2}
\mathcal{B}_{\text{BMH}}(r)=A_{ij}\, \exp\left(\dfrac{\sigma_{ij}-r}{\rho}\right)-\dfrac{C_{ij}}{r^6}-\dfrac{D_{ij}}{r^8}.
\end{equation}
Then the short-range force exerted on the $i$th particle is expressed by
\begin{equation}\label{eq::Fshort}
\bm{F}_{\mathcal{N},i} =\sum_{j\in\mathbb{I}(i)}\left[q_iq_j\left( \dfrac{1}{r_{ij}^3}-\sum_{m=0}^M\frac{w_m}{b^{2m}s^2}e^{-r_{ij}^2/(2b^{2m}s^2)}  \right)+\dfrac{A_{ij}}{\rho r_{ij}}e^{\frac{\sigma_{ij}-r_{ij}}{\rho}}-\dfrac{6C_{ij}}{r_{ij}^8}-\dfrac{8D_{ij}}{r_{ij}^{10}}\right]\bm{r}_{ij}.
\end{equation}
After this split, the short-range part can be truncated at a cutoff \(r_c\), while the smooth long-range part is evaluated in Fourier space.

Let us define the Fourier transform pairs as 
\begin{equation}
    \widehat{f}(\bm{k}):=\int_{\Omega}f(\bm{r})e^{-i\bm{k}\cdot\bm{r}}d\bm{r}
\end{equation}
and
\begin{equation}
f(\bm{r})=\frac{1}{V}\sum_{\bm{k}}\widehat{f}(\bm{k})e^{i\bm{k}\cdot\bm{r}},
\end{equation}
where $\bm{k}=2\pi(m_x/L_x,m_y/L_y,m_z/L_z)$ and $\bm{m}=(m_x,m_y,m_z)\in\mathbb{Z}^3$. By applying the Poisson summation formula to Eq.~\eqref{eq:UfBMH}, the long-range BMH energy can be split into three parts:
\begin{equation}\label{eq::UF}
    U_{\mathcal{F}}=	U_{\mathcal{F}}^{\text{*}}+	U_{\mathcal{F}}^{\text{0}}-	U_{\mathcal{F}}^{\text{self}},
\end{equation}
where
\begin{equation}
U_{\mathcal{F}}^{\text{*}}=\sum_{|\bm{k}|\neq0}\widehat{\mathcal{F}}_b^{s}(\bm{k})\frac{|\rho(\bm{k})|^2}{2V}  
\end{equation}
is the contribution from non-zero modes,
\begin{equation}
U_{\mathcal{F}}^{\text{0}}=\lim_{\bm{k}\to0}\widehat{\mathcal{F}}_b^{s}(\bm{k})\frac{|\rho(\bm{k})|^2}{2V}
\end{equation}
is the zero-mode contribution, and the last term $U_{\mathcal{F}}^{\mathrm{self}}$ is added to remove the unwanted self interaction. Here, 
\begin{equation}\label{eq:Fouriertransform}
\widehat{\mathcal{F}}^{s}_b(\bm{k}) = (2\pi)^{3/2}\sum_{m=0}^{M} w_m b^{3m} s^{3}e^{-b^{2m} s^{2} k^{2}/2}
\end{equation}
is the Fourier transform of $\mathcal{F}^{s}_b(r)$ with $k=|\bm{k}|$, and 
\begin{equation}
\rho(\bm{k}):=\sum_{i=1}^Nq_ie^{i\bm{k}\cdot\bm{r}_i}
\end{equation}
denotes the charge structure factor. In this paper, we assume tinfoil boundary conditions so that $U_{\mathcal{F}}^{0}\equiv0$. More detailed discussions regarding the zero-mode term can be found in the literature~\cite{liang2023random}. The long-range force acting on the $i$th particle is given by 
\begin{equation}\label{eq::Long}
\bm{F}_{\mathcal{F},i} = -\nabla_{\bm{r}_i}U_{\mathcal{F}}=-\sum_{\bm{k}\neq 0} \frac{q_i\bm{k}}{V} 
    \widehat{\mathcal{F}}_b^{s}(\bm{k}) \text{Im} \left( e^{-i\bm{k} \cdot \bm{r}_i} \rho(\bm{k}) \right).
\end{equation}
Note that the last term in Eq.~\eqref{eq::UF} has no contribution to the force, i.e. $\nabla_{\bm{r}_i}U_{\mathcal{F}}^{\text{self}}=0$, since the particle number is typically invariant in the MD simulations. 


\subsection{The IRBSOG method}\label{sec:II.C}
In this section, we introduce the improved random-batch sum-of-Gaussians (IRBSOG) method to simplify the calculation of the BMH potential. Based on the SOG-split formulation, the short-range and long-range forces are treated separately using tailored random-batch type algorithms. 


\begin{figure}[!ht]
\centering	
\includegraphics[width=0.65\linewidth]{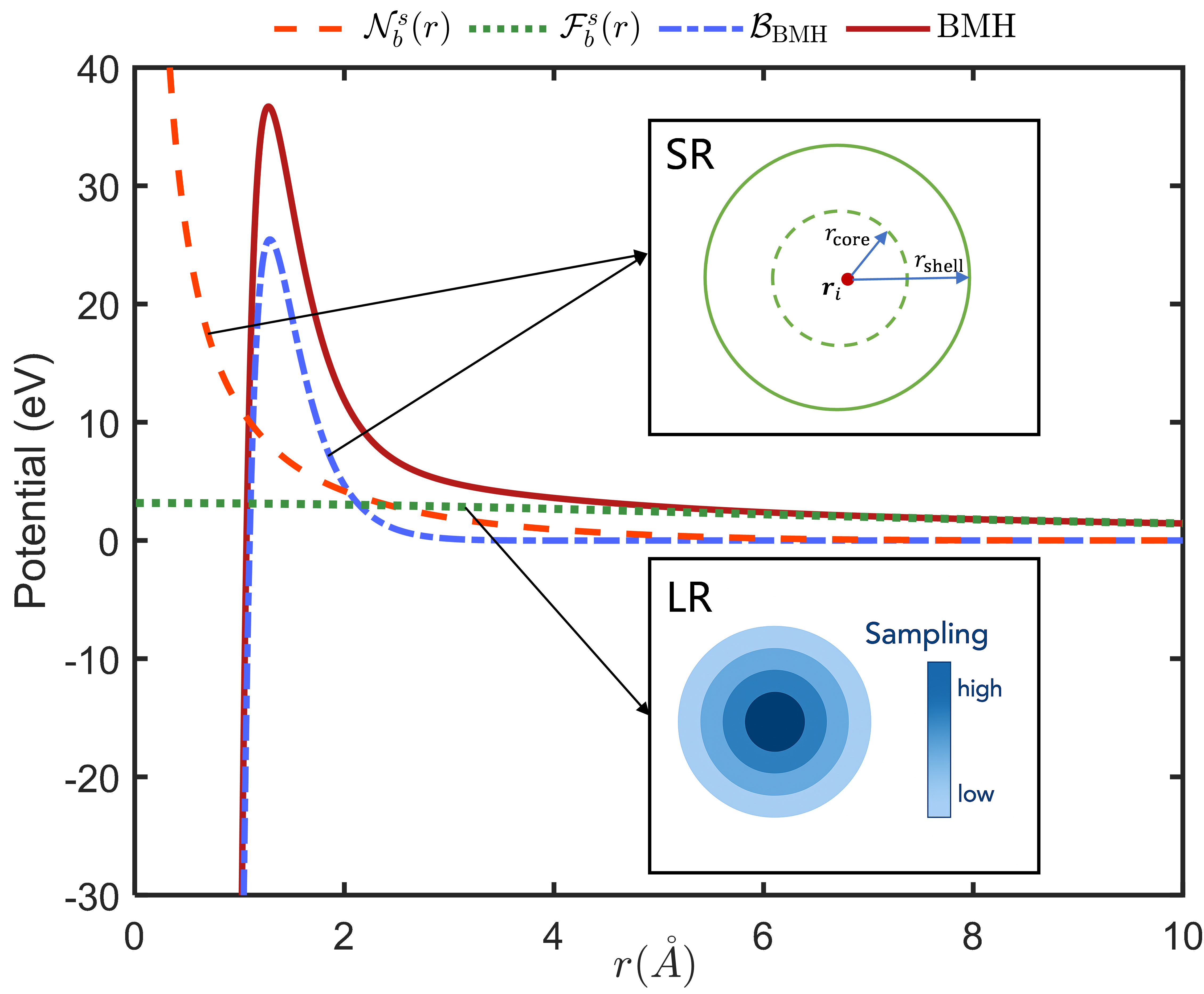}
\captionsetup{justification=raggedright, singlelinecheck=false}
\caption{Schematic of the IRBSOG method and the Born-Mayer-Huggins (BMH) potential decomposition. The total BMH potential (red solid line) is decomposed into components handled by distinct algorithms. The Coulomb potential is split into a short-range part $\mathcal{N}_{b}^{s}(r)$ (orange dashed line) and a long-range part $\mathcal{F}_{b}^{s}(r)$ (green dot-dashed line). The short-range non-Coulombic BMH terms $\mathcal{B}_{\text{BMH}}(r)$ (blue dot-dashed line) is also illustrated. The upper inset shows the core-shell structure of short-range and the lower one shows the importance sampling of long-range.}
\label{fig::decomposition}
\end{figure}

To accelerate the short-range force $\bm{F}_{\mathcal{N},i}=-\nabla_{\bm{r}_i}U_{\mathcal{N}}^{\text{BMH}}$, we adopt the RBL framework~\cite{liang2021random} with a dual-cutoff scheme. As illustrated in the inset of Fig.~\ref{fig::decomposition}, the neighbor list of each particle is split into a core-shell structure with radii $r_{\text{core}}<r_{\text{shell}}$. The short-range force is decomposed as
$\bm{F}_{\mathcal{N},i}=\bm{F}_{i,\text{core}}+\bm{F}_{i,\text{shell}}$, where the core contribution is computed by direct summation:
\begin{equation}
\bm{F}_{i,\text{core}}=\sum_{r_{ij}<r_{\text{core}}}\bm{f}_{ij},
\end{equation}
where $\bm{f}_{ij}$ is the force exerted by particle $j$ on particle $i$. In the shell region, we randomly sample a mini-batch $B(i)$ of $P_R$ neighbors of particle $i$ to construct a stochastic estimator:
\begin{equation}
\bm{F}_{i,\text{shell}}^*=\frac{N_{\text{shell}}}{P_R}\sum_{j\in B(i)}\bm{f}_{ij},
\end{equation}
where $N_{\text{shell}}$ is the number of neighbors with $r_{\text{core}}<r_{ij}<r_{\text{shell}}$ (and $\bm{f}_{ij}=0$ otherwise). A momentum correction term $\bm{F}^{\text{corr}}$ is applied to ensure the conservation of total momentum:
\begin{equation}
\bm{F}^{\text{corr}}=\frac{1}{N}\sum_{i}^N(\bm{F}_{i,\text{core}}+\bm{F}^*_{i,\text{shell}}),
\end{equation}
which is a random
variable with zero mean and bounded variance. The final approximation for the short-range force is
\begin{equation}
\bm{F}_{\mathcal{N},i}^*=\bm{F}_{i,\text{core}}+\bm{F}^*_{i,\text{shell}}-\bm{F}^{\text{corr}}.
\end{equation}
This RBL scheme is essential for overcoming the memory-bound bottleneck of classical neighbor-list algorithms: it avoids scanning all neighbors in the shell region and instead evaluates forces for only a small subset, greatly reducing the computational and memory cost with minimal impact on accuracy.

The long-range force $\bm{F}_{\mathcal{F},i}$ is evaluated in Fourier space using the RBSOG idea~\cite{liang2023random}. Starting from its exact spectral representation (Eq.~\eqref{eq::Long}), the Fourier transform $\widehat{\mathcal{F}}_{b}^{s}(\bm{k})$ is still an SOG, which is summable and can be normalized as a discrete probability distribution. One can regard the sum as an expectation over the distribution
\begin{equation}\label{eq::Pk}
\mathscr{P}(\bm{k}):=\frac{k^2\widehat{\mathcal{F}}_b^{s}(\bm{k})}{S},
\end{equation}
where the sum is defined by $S:=\sum_{\bm{k}\neq0}k^2\widehat{\mathcal{F}}_b^{s}(\bm{k}).$  The distribution Eq.~\eqref{eq::Pk} is an SOG scaled by $k^2$, and can be efficiently sampled by the Metropolis-Hastings (MH) algorithm~\cite{metropolis1953equation}; see~\cite{liang2023random} for details. At each MD step, a fixed mini-batch $P_F$ is picked and sampled frequencies $\{\bm{k}_p,p=1,\dots,P_F\}$ are drawn from $\mathscr{P}(\bm{k})$. The long-range force is then approximated as an expectation over $\mathscr{P}(\bm{k})$, 
\begin{equation}\label{eq::long-rangeforce}
\bm{F}_{\mathcal{F},i}^*=-\frac{Sq_i}{P_FV}\sum_{p=1}^{P_F}\frac{\bm{k}_p\operatorname{Im}\left(e^{-i \boldsymbol{k_p} \cdot \boldsymbol{r}_i} \rho(\boldsymbol{k_p})\right)}{\vert\bm{k}_p\vert^2}.
\end{equation} 
This replaces FFT by a stochastic approximation, resulting in a cheaper version compared to lattice-based Ewald-type methods.

In practice, we use the stochastic estimators $\bm{F}_{\mathcal{N},i}^{*}$ and $\bm{F}_{\mathcal{F},i}^*$ for the short- and long-range forces $\bm{F}_{\mathcal{N},i}$ and $\bm{F}_{\mathcal{F},i}$, respectively.  This combined approach leverages the strengths of both batch methods: the RBL efficiently handles the complex, consolidated short-range interactions in real space, while the RBSOG efficiently handles the smooth, long-range interactions in Fourier space. Their synergistic integration within the SOG-split framework is the foundation of our unified $O(N)$ algorithm. A sketch map of the resulting IRBSOG method is shown in Fig.~\ref{fig::decomposition}.

\subsection{Convergence analysis of IRBSOG-based simulations}
We analyze the convergence of the IRBSOG-based MD simulations. Let $\bm{F}_i^{\mathrm{total}}$ be the exact total
nonbonded force on particle $i$ under the full BMH potential, and let
$\bm{F}_i^{*}=\bm{F}_{\mathcal{N},i}^{*}+\bm{F}_{\mathcal{F},i}^{*}$ be its stochastic approximation produced by IRBSOG.
We define the force fluctuation \begin{equation}
\bm{\chi}_i:= \bm{F}_i^{\text{total}}-\bm{F}_i^*.
\end{equation}
Since the random mini-batches used in real space (RBL) and Fourier space (RBSOG) are
independent, the covariance term vanishes and we have
\begin{equation}
    \text{var}[\bm{\chi}_i]=\text{var}[\bm{F}_{\mathcal{N},i}^{*}]+\text{var}[\bm{F}_{\mathcal{F},i}^{*}].
\end{equation}
The variance of the RBSOG estimator $\bm{F}_{\mathcal{F},i}^{*}$ has been analyzed in Ref.~\cite{liang2023random};
we restate the main bound in Lemma~\ref{lemma:var_long}. Theorem~\ref{thm::consistent} then combines the
Fourier-space bound with a novel real-space bound for RBL, showing that $\bm{F}_i^{*}$ is an unbiased estimator of
$\bm{F}_i^{\mathrm{total}}$ and that the variance decays as the batch sizes increase.



\begin{lemma}\label{lemma:var_long}
Under the Debye-H\"uckel (DH) assumption, the variance of the long-range force approximation $\bm{F}_{\mathcal{F},i}^{*}$ satisfies
\begin{equation}
\operatorname{var}[\bm{F}_{\mathcal{F},i}^{*}]
\le
\frac{S C q_i^2}{P_F V}\sum_{m=0}^M w_m,
\end{equation}
where $C$ is the constant in the DH bound, $q_i$ is the charge of particle $i$, $V$ is the system volume, $P_F$ is
the Fourier-space batch size, and $w_m$ are the Gaussian weights used in Eq.~(7). Consequently,
$\operatorname{var}[\bm{F}_{\mathcal{F},i}^{*}] = O(1/P_F)$, and this bound is independent of the system size $N$
and the number of Gaussians $M$.
\end{lemma}

\begin{theorem}\label{thm::consistent}
Let the number density $\rho_r=N/V$ be fixed. The force fluctuation $\bm{\chi}_i$ satisfies
\begin{equation}
\mathbb{E}\bm{\chi}_i=\bm{0}.
\end{equation}
Moreover, under the DH assumption,
\begin{equation}
\operatorname{var}[\bm{\chi}_i]=O\!\left(\frac{1}{P_F}+\frac{1}{P_R}\right),
\end{equation}
where $P_F$ and $P_R$ are the batch sizes used in Fourier space and real space calculations, respectively.
\end{theorem}

\begin{proof}

The proof can be obtained by extending the random-batch analysis in Refs.~\cite{jin2021random,liang2021random}. We first note that
$\mathbb{E}\bm{\chi}_i=\bm{0}$ holds because both $\bm{F}_{\mathcal{N},i}^{*}$ and $\bm{F}_{\mathcal{F},i}^{*}$ are
constructed as unbiased estimators of their corresponding exact contributions. For the variance, Lemma~\ref{lemma:var_long} gives
\begin{equation}
\operatorname{var}[\bm{F}_{\mathcal{F},i}^{*}] = O(1/P_F).
\end{equation}
It remains to bound $\operatorname{var}[\bm{F}_{\mathcal{N},i}^{*}]$ for the RBL estimator.

We split the
short-range variance into the SOG residual, the exponential repulsion, and the dispersion contributions:
\begin{equation}\label{eq::variance}
    \text{var}[\bm{F}_{\mathcal{N},i}^{*}]\le\text{var}[\bm{F}_{\text{Coulomb},i}^{*}]+\text{var}[\bm{F}_{\text{rep},i}^{*}]+\text{var}[\bm{F}_{\text{disp},i}^{*}].
\end{equation}
We first consider the variance of the short‑range SOG residual force. Starting from the expression for the approximate force, we obtain
\begin{equation}~\label{varcoulomb0}
\begin{aligned}
        \text{var}[\bm{F}_{\text{Coulomb},i}^{*}]&\lesssim\dfrac{(N_{\text{shell}}-P_R)}{P_R}\int_{r_{\text{core}}}^{r_{\text{shell}}}4\pi\rho_r r^2\left|\dfrac{1}{r^2}-\sum_{m=0}^M\frac{w_mr}{b^{2m}s^2}e^{-\frac{r^2}{2b^{2m}s^2}}\right|^2dr\\
    &=\dfrac{(N_{\text{shell}}-P_R)}{P_R}\int_{r_{\text{core}}}^{r_{\text{shell}}}\dfrac{4\pi\rho_r}{r^2}\left|1-\dfrac{\sqrt{2}r^3\ln b}{\sqrt{\pi}s^3}\sum_{m=0}^M\dfrac{1}{b^{3m}}e^{-\frac{r^2}{2b^{2m}s^2}}\right|^2dr.
    \end{aligned}
\end{equation}
For the truncated series, we decompose the error into the bilateral-series error and the positive/negative tails:
\begin{equation}
\begin{aligned}
        \left|1-\dfrac{\sqrt{2}r^3\ln b}{\sqrt{\pi}s^3}\sum_{m=0}^{M}\dfrac{1}{b^{3m}}e^{-\frac{r^2}{2b^{2m}s^2}}\right| &\le \left|1-\dfrac{\sqrt{2}r^3\ln b}{\sqrt{\pi}s^3}\sum_{m=-\infty}^{\infty}\dfrac{1}{b^{3m}}e^{-\frac{r^2}{2b^{2m}s^2}}\right|+ \left|\dfrac{\sqrt{2}r^3\ln b}{\sqrt{\pi}s^3}\sum_{m=M+1}^{\infty}\dfrac{1}{b^{3m}}e^{-\frac{r^2}{2b^{2m}s^2}}\right| \\
        &\quad + \left|\dfrac{\sqrt{2}r^3\ln b}{\sqrt{\pi}s^3}\sum_{m=-\infty}^{-1}\dfrac{1}{b^{3m}}e^{-\frac{r^2}{2b^{2m}s^2}}\right|.
\end{aligned}
\end{equation}
The first term is the relative error of BSA for the $1/r^3$ kernel, which has a uniform error bound as
$b\to 1$ similar to Eq.~\eqref{eq::asymperror}~\cite{predescu2020u}:
\begin{equation}\label{bsa}
    \left|1-\dfrac{\sqrt{2}r^3\ln b}{\sqrt{\pi}s^3}\sum_{m=-\infty}^{\infty}\dfrac{1}{b^{3m}}e^{-\frac{r^2}{2b^{2m}s^2}}\right|\lesssim 4\sqrt{2}\left(\dfrac{9}{4}+\dfrac{\pi^2}{(\ln b)^2}\right)^{1/2}e^{-\frac{\pi^2}{2\ln b}}.
\end{equation}
The last two terms represent the truncation errors from the positive and negative tails of the BSA, respectively, and can be bounded as follows. For the positive tail, using $e^{-\frac{r^2}{2b^{2m}s^2}}\le1$ gives
\begin{equation}~\label{positive}
\left|\dfrac{\sqrt{2}r^3\ln b}{\sqrt{\pi}s^3}\sum_{m=M+1}^{\infty}\dfrac{1}{b^{3m}}e^{-\frac{r^2}{2b^{2m}s^2}}\right| \le \dfrac{\sqrt{2}r^3\ln b}{\sqrt{\pi}s^3} \cdot \frac{b^{-3M}}{b^3-1}.
\end{equation}
For the negative tail, letting $j=-m\ge1$, we have
\begin{equation}
\left|\dfrac{\sqrt{2}r^3\ln b}{\sqrt{\pi}s^3}\sum_{m=-\infty}^{-1}\dfrac{1}{b^{3m}}e^{-\frac{r^2}{2b^{2m}s^2}}\right| = \dfrac{\sqrt{2}r^3\ln b}{\sqrt{\pi}s^3}\sum_{j=1}^{\infty} b^{3j} e^{-\frac{r^2 b^{2j}}{2 s^2}}.
\end{equation}
This series decays doubly exponentially in $j$. A convenient bound is obtained by isolating the first term and
bounding the remainder via
\begin{equation}
\left|\dfrac{\sqrt{2}r^3\ln b}{\sqrt{\pi}s^3}\sum_{m=-\infty}^{-1}\dfrac{1}{b^{3m}}e^{-\frac{r^2}{2b^{2m}s^2}}\right|  \le \dfrac{\sqrt{2}r^3\ln b}{\sqrt{\pi}s^3} \left( b^3 e^{-\frac{r^2 b^2}{2 s^2}} + \frac{1}{2\ln b} \left(\frac{2s^2}{r^2}\right)^{3/2} \Gamma\!\left(\frac{3}{2}, \frac{r^2 b^2}{2 s^2}\right) \right),
\end{equation}
where $\Gamma(s,x)=\int_{x}^{\infty}t^{s-1}e^{-t}dt$ is the upper incomplete Gamma function. For typical values of $r$ ($r\ge r_{\text{core}}>0$) and $b>1$, the term $r^2b^2/2s^2$ is large and the incomplete Gamma function admits the asymptotic approximation $\Gamma(3/2,r^2b^2/2s^2)\approx\sqrt{r^2b^2/2s^2}e^{-r^2b^2/2s^2}$. Consequently, the negative tail bound decays exponentially with $r^2$ and is negligible in the integration over $[r_{\text{core}},r_{\text{shell}}]$.
Substituting the bounds in Eq.~\eqref{bsa} and \eqref{positive} into the variance integral yields
\begin{equation}~\label{varcoulomb}
\begin{aligned}
    \text{var}[\bm{F}_{\text{Coulomb},i}^{*}] &\lesssim\dfrac{(N_{\text{shell}}-P_R)}{P_R}\int_{r_{\text{core}}}^{r_{\text{shell}}}\dfrac{4\pi\rho_r}{r^2}\left|1-\dfrac{\sqrt{2}r^3\ln b}{\sqrt{\pi}s^3}\sum_{m=0}^M\dfrac{1}{b^{3m}}e^{-\frac{r^2}{2b^{2m}s^2}}\right|^2dr\\
    &\lesssim\dfrac{(N_{\text{shell}}-P_R)}{P_R}\int_{r_{\text{core}}}^{r_{\text{shell}}}\dfrac{4\pi\rho_r}{r^2}\cdot2\left[\left(72+\dfrac{32\pi^2}{(\ln b)^2}\right)e^{-\frac{\pi^2}{\ln b}}+\frac{2(\ln b)^2}{\pi s^6(b^3-1)^2}b^{-6M}r^6\right]dr
    \\  
    &=\dfrac{(N_{\text{shell}}-P_R)}{P_R}8\pi\rho_r\Bigg[\left(72+\dfrac{32\pi^2}{(\ln b)^2}\right)e^{-\frac{\pi^2}{\ln b}}\int_{r_{\text{core}}}^{r_{\text{shell}}}r^{-2}dr+\frac{2(\ln b)^2}{\pi s^6(b^3-1)^2}b^{-6M}\int_{r_{\text{core}}}^{r_{\text{shell}}}r^4dr\Bigg]\\
     &=\dfrac{(N_{\text{shell}}-P_R)}{P_R}8\pi\rho_r\Bigg[\left(72+\dfrac{32\pi^2}{(\ln b)^2}\right)e^{-\frac{\pi^2}{\ln b}}\left({\dfrac{1}{r_{\text{core}}}}-\dfrac{1}{r_{\text{shell}}}\right)+\frac{2(\ln b)^2}{5\pi s^6(b^3-1)^2}\left(r_{\text{shell}}^5-r_{\text{core}}^5\right)b^{-6M}\Bigg].
    \end{aligned}
\end{equation}
Then considering the variance of the exponential repulsion force, one can obtain:
\begin{equation}\label{varrep}
        \text{var}[\bm{F}_{\text{rep},i}^{*}]\lesssim\dfrac{(N_{\text{shell}}-P_R)}{P_R}\int_{r_{\text{core}}}^{r_{\text{shell}}}4\pi \rho_rr^2 \left(\dfrac{A_{ij}}{\rho}\right)^2e^{2\sigma_{ij}/\rho}e^{-2r/\rho}dr 
        \lesssim\dfrac{C_1(N_{\text{shell}}-P_R)}{P_R}\rho_{r}e^{-2r_{\text{core}}/{\rho}},
\end{equation}
where $C_1$ is a constant depending only on the BMH parameters. Finally, the variance of the dispersion force can be bounded by the following:
\begin{equation}\label{vardisp}
\begin{aligned}
    \text{var}[\bm{F}_{\text{disp},i}^{*}]\lesssim\dfrac{(N_{\text{shell}}-P_R)}{P_R}\int_{r_{\text{core}}}^{r_{\text{shell}}}4\pi \rho_rr^2 \left(\dfrac{6C_{ij}}{r^7}+\dfrac{8D_{ij}}{r^{9}}\right)^2dr 
        \lesssim\dfrac{C_2(N_{\text{shell}}-P_R)}{P_R}\rho_{r}r_{\text{core}}^{-11},
    \end{aligned}
\end{equation}
where $C_2$ is another constant depending only on the BMH parameters.

Since the constants $C$, $C_1$, and $C_2$ are all independent of $N$ and $M$, combining the results of the short-range force approximation in Eqs.~\eqref{varcoulomb}, \eqref{varrep} and \eqref{vardisp}, we can conclude that the variance of $\bm{F}_{\mathcal{N},i}$ scales as $O(1/P_R)$ and we complete the proof.
\end{proof}

We now analyze the convergence of IRBSOG-based MD. Let $\Delta t$ be the step size of MD time integration. Under the canonical (NVT) ensemble, a thermostat is generally used to sample correct ensemble distribution, such as the Langevin thermostat~\cite{frenkel2023understanding} and the Nos\'e-Hoover (NH) thermostat~\cite{hoover1985canonical}. We consider the Langevin thermostat for convenience in analysis. Let $(\bm{r}_i,\bm{v}_i)$ be the solution of the underdamped Langevin equations of motion,
\begin{equation}
    \begin{aligned}
        &d\bm{r}_i=\bm{v}_i\ dt,\\
        &m\ d\bm{v}_i=(\bm{F}^{\text{total}}_i-\gamma\bm{v}_i)\ dt+\sqrt{2\gamma k_BT}\ d\bm{W}_i,
    \end{aligned}
\end{equation}
where $k_B$ is the Boltzmann constant, $\gamma$ is the friction coefficient, and ${\bm{W}_i}$ are independently identically distributed Wiener processes. Let $(\bm{r_i^*, \bm{v}_i^*})$ be the solutions with the force $\bm{F}^{\text{total}}_i$ approximated by $\bm{F}_i^*$, the force calculated by IRBSOG, with the same initial data. The following Theorem indicates that the IRBSOG accurately captures finite-time dynamics despite the random batch approximation of the force under the Langevin thermostat. The proof of this theorem can be obtained by following previous work~\cite{jin2020random,jin2021convergence}, and we will not present it here.
\begin{theorem}\label{thm::1}
	 Suppose that the forces $\bm{F}^{\text{total}}_i$ are bounded and Lipschitz and $\mathbb{E}\chi_i=0$. Under the synchronization coupling assumption that the same initial values as well as the same Wiener process $\bm{W}_i$ are used, then for any finite $T>0$, there exists $C(T)>0$ such that
	 \begin{equation}\label{eq::theorem}
     \sup_{t\in[0,T]}\sqrt{\dfrac{1}{N}\sum_{i=1}^N\mathbb{E}\left(|\bm{r_i-\bm{r}^*_i}|^2+|\bm{v_i-\bm{v}^*_i}|^2\right)}\le C(T)\sqrt{\Lambda\Delta t},
	 \end{equation}  
where $\Lambda=\max_i(\mathbb{E}|\chi_i|^2)$ is an upper bound for the variance of the random approximation.
\end{theorem}
By Theorem~\ref{thm::consistent}, we have $\Lambda=O(1/P_R+1/P_F)$. For long-time simulations, a uniform-in-time error estimate can be derived in a manner similar to~\cite{jin2022random} for systems with regular force fields. We note that Theorem~\ref{thm::1} assumes the force to be bounded and Lipschitz continuous, whereas the Coulomb kernel has a singularity at the origin. Nevertheless, we expect that the error bound in Eq.~\eqref{eq::theorem} remains valid in practice. This is because the BMH potential includes both Lennard-Jones-type interactions and a strong short-range exponential repulsion, which prevent particles from approaching arbitrarily close, so that the singularity does not occur during the simulation. A fully rigorous justification of this argument remains difficult and highly nontrivial. For simulations in the microcanonical (NVE) ensemble, random-batch-type methods can be combined with a weakly coupled thermal bath to improve long-time energy stability~\cite{liang2024energy}. For further discussion of other thermostats and ensembles, we refer the reader to~\cite{jin2021random,liang2023random}.

\subsection{Implementation details and complexity analysis}\label{sec:II.D}
We have implemented the IRBSOG method in a high-performance MD code that uses distributed memory parallelism and vectorization. The code is built on the Message Passing Interface (MPI) and AVX-512 single-instruction multiple-data (SIMD) instructions. For short-range forces in the core region, we set an additional inner cutoff radius $r_{\text{in}}<r_{\text{core}}$. For $0<r<r_{\text{in}}$, forces are computed using Taylor expansions; for $r_{\text{in}}<r<r_{\text{core}}$, we use a bitmask-based table-lookup technique~\cite{wolff1999tabulated} with linear interpolation between tabulated points. For the long-range approximation in Eq.~\eqref{eq::long-rangeforce}, the sampled Fourier modes and particle positions are vectorized when evaluating the structure factors $\rho(\bm{k})$ from local atoms on each MPI rank. A single global all-reduce then assembles the complete structure factors. The approximate force $\bm{F}_{\mathcal{F},i}^{*}$ for each particle is then computed from these structure factors. For better demonstration of the performance of IRBSOG, other LAMMPS components, including domain decomposition, ghost-atom communication, and neighbor lists construction with the multiple-page data structure, remain unchanged.

We analyze the per-step computational complexity and memory use of IRBSOG. Without loss of generality, assume a uniform particle density $\rho_{r}=N/V$. For the long-range part, with a fixed Fourier mini-batch size $P_{F} = O(1)$, the computational cost is $O(P_FN)$ and the memory storage is $O(P_{F})$ (the number of Fourier modes involved in the calculation). By comparison, FFT-based methods require $O(N\log N)$ work and $O(N)$ memory usage. For the short-range part, the per-particle neighbor list size (and both complexity and memory usage) required in the classical cutoff approach scales as $O(4\pi r_{\text{c}}^3\rho_r/3)$. IRBSOG reduces this to $O(4\pi r_{\text{core}}^3\rho_r/3+P_R)$, where $P_R$ is the shell mini-batch size. If one safely adopts $r_c=r_{\text{shell}}$, $r_{\text{shell}}/r_{\text{core}} = 2$ and an appropriate $P_R$, both the memory usage and the runtime for the short-range part are reduced by about one order of magnitude. Since short-range calculations are often memory-bound, IRBSOG is expected to be more efficient than the original RBSOG, in which only the long-range part is accelerated. The above analysis indicates that IRBSOG achieves an overall linear computational complexity.

\section{Numerical Examples}\label{sec:III}
In this section, we perform all-atom simulations with the IRBSOG-based MD under the NVT ensemble to validate both the accuracy and efficiency of the proposed method, with two benchmark systems including the molten NaCl and the molten alkali binary halide systems. All the simulations were conducted by our method implemented in the LAMMPS~\cite{thompson2022lammps} (version 21Nov2023), and were performed on the Siyuan Mark-I cluster at Shanghai Jiao Tong University, which comprises 936 nodes with 2 Intel Xeon ICX Platinum 8358 CPU (26 GHz, 32 cores) and 512 GB memory per node.

\subsection{Accuracy comparison in molten NaCl systems}\label{section:A}

We compare IRBSOG and PPPM for molten NaCl systems governed by the BMH potential. Simulations use a cubic box with 8192 ions (4096 Na$^{+}$ and 4096 Cl$^{-}$). After a $50~ps$ equilibration in the isothermal-isobaric (NPT) ensemble at $0$ $GPa$ and $1100$ $K$, we perform $1~ns$ production run in the NVT ensemble with a Nos\'{e}-Hoover thermostat~\cite{hoover1985canonical}, saving data every $100$ steps. The time step is $\Delta t=1~fs$. For PPPM, we set the Ewald splitting parameter $\alpha =0.09$, yielding a predicted relative error of $10^{-5}$. Short-range interactions are truncated at $r_c=10~\mathring{A}$. The BMH parameters provided by Tosi and Fumi~\cite{fumi1964ionic,tosi1964ionic} are used throughout (summarized in Table~\ref{tab1}).

\begin{table}[H]
	\caption{BMH potential parameters for molten NaCl.}
		\setlength{\tabcolsep}{8pt}
	\renewcommand{\arraystretch}{1.3}
	\centering
	\begin{tabular}{cccc}
		\hline
		\  
		&Na-Na &Na-Cl &Cl-Cl \\ \hline
		A$/$eV    & 0.303   & 0.242  &  0.182 \\ \hline
		$\sigma$$/$$\mathring{A}$    &2.340 &2.755 &3.170 \\ \hline
		C$/$eV$\cdot\mathring{A}^6$ 	   &1.049  &6.991  &72.40 \\  \hline
		D$/$eV$\cdot\mathring{A}^8$  	 &0.499  &8.676 &  145.427\\  \hline
		$\rho$$/$$\mathring{A}$ 	    &0.317  &0.317  &0.317 \\  \hline
	\end{tabular}
	\label{tab1}
\end{table}

We evaluate accuracy using the mean-square displacement (MSD) and radial distribution function (RDF). For IRBSOG, we set $r_{\text{core}}=5~\mathring{A}$ and $r_{\text{shell}}=10~\mathring{A}$, giving about $120$ particles in the shell region per ion. Figs.~\ref{rdf1} (a-b) show the MSD of Na and the Na-Na RDF, respectively. We test four combinations of batch sizes for IRBSOG: $(P_R,P_F)=(20,50)$, $(30,50)$, $(20,100)$, and $(30,100)$. As $P_R$ and $P_F$ increase, the MSD and RDF generated by the IRBSOG converge to those of PPPM; with $P_R=30$ and $P_F=100$, results of IRBSOG and PPPM are statistically indistinguishable. The agreement holds across multiple time scales, showing that IRBSOG reproduces the dynamics of system well.

\begin{figure*}[!ht]
	\centering	
         \includegraphics[width=0.8\linewidth]{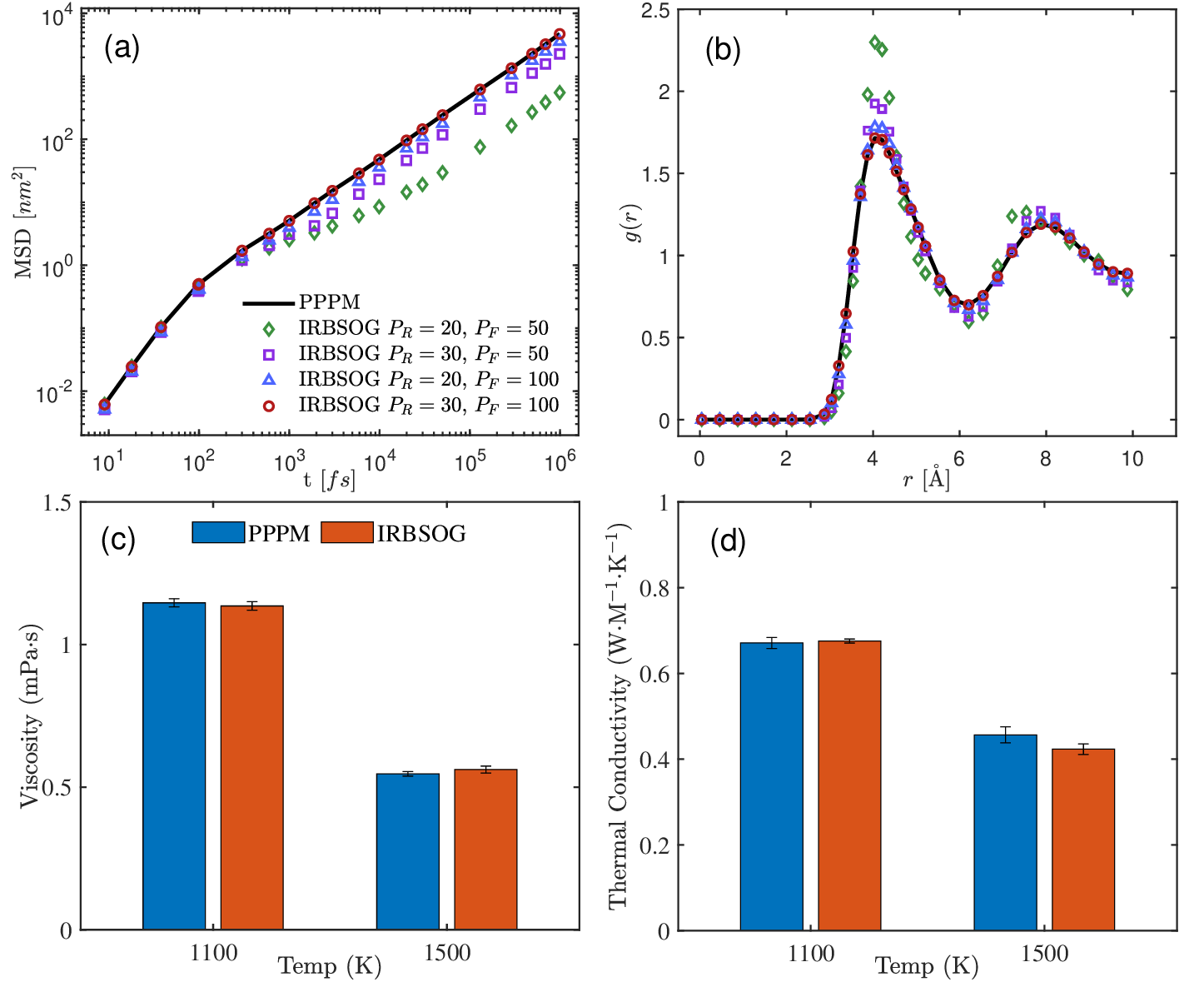}
                    \captionsetup{justification=raggedright, singlelinecheck=false}
	\caption{(a) Na MSD and (b) Na-Na RDF of molten NaCl systems governed by the BMH potential. IRBSOG results with shell-region batch size $P_R\in\{20, 30\}$ and Fourier-space batch size $P_F\in\{50, 100\}$ are compared with PPPM. Panels (c) and (d) compare viscosity and thermal conductivity of NaCl at $1100~K$ and $1500~K$: IRBSOG with $P_R=30$, $P_F=100$ (orange) versus PPPM (blue).}
	\label{rdf1}
\end{figure*}

We further assess accuracy in fluid mechanics and thermodynamics by computing the viscosity and thermal conductivity. Viscosity is the internal resistance to flow and is obtained from the Green-Kubo relation in equilibrium~\cite{dunweg2008colloidal}:
\begin{equation}
    \eta=\lim_{\tau\to\infty}\frac{1}{k_BTV}\int_0^{\tau}\langle S_{xy}(t)\cdot S_{xy}(0)\rangle dt,
\end{equation}
where $V$ is the system volume and $T$ the temperature. The $xy$-component of the stress tensor, $S_{xy}(t)$, is calculated by
\begin{equation}
S_{xy}=\sum^{N}_{i=1}\left[m_iv_{xi}v_{yi}+\frac{1}{2}\sum_{j\neq i}x_{ij}f_y(r_{ij})\right],
\end{equation}
with $m_i$ the mass of ion $i$, $v_{xi}$ and $v_{yi}$ its velocity components, $x_{ij}$ the $x$-component of $\bm{r}_{ij}=\bm{r}_i-\bm{r}_{j}$, and $f_y(r_{ij})$ the $y$-component of the pair force $\bm{f}_{ij}$ on ion $i$ from ion $j$. Thermal conductivity $\lambda$ is also computed via the Green-Kubo formula:
\begin{equation}
\lambda=\lim_{\tau\to\infty}\frac{V}{3k_BT^2}\int_0^{\tau}\langle\bm{J}(0)\cdot \bm{J}(t)\rangle dt,
\end{equation}
where the heat flux $\bm{J}(t)$ is 
\begingroup
\small
\begin{equation}
\bm{J}(t)=\frac{1}{V}\left[\sum_{i=1}^N\bar{e}_i\bm{v}_i(t)+\frac{1}{2}\sum_{\substack{i,j=1 \\ i<j}}^N\left(\bm{F}_{ij}(t)\cdot(\bm{v}_i(t)+\bm{v}_j(t)\right)\bm{r}_{ij}(t)
 \right],
\end{equation}
\endgroup
with $\bar{e}_i$ the per-atom energy. As shown in Fig.~\ref{rdf1}(c-d), IRBSOG with $P_R=30$ and $P_F = 100$ yields viscosity and thermal conductivity for molten NaCl at $1100~K$ and $1500~K$ that are statistically indistinguishable from PPPM with $10^{-5}$ accuracy. This indicates that a modest batch size suffices for IRBSOG to produce accurate MD results. For reference, IRBE~\cite{liangIRBE} requires a Fourier-space batch size of $P_F=200$, twice that of the IRBSOG in simulating BMH systems, when simulating pure water and ionic liquid systems. The factor of two speedup comes from replacing the Ewald decomposition with the SOG decomposition.

\subsection{Time performance in large-scale simulations}
We compare PPPM, RBSOG and IRBSOG in LAMMPS on the molten NaCl system described above. To ensure a fair comparison, we set the target relative force error to $10^{-5}$ and the real-space cutoff to $10~\mathring{A}$ for all three methods. For RBSOG, we use $P_F=100$ as the Fourier-space batch size. For IRBSOG, we set $r_{\text{core}}=5~\mathring{\text{A}}$, $r_{\text{shell}}=10~\mathring{\text{A}}$, and batch sizes $P_R=30$ and $P_F=100$, as validated in Section~\ref{section:A}. Each run is $10000$ steps, and we report the average CPU time per step. Fig.~\ref{linear} shows the results as a function of the number of particles, using $1000$ CPU cores. All three methods show the expected closed-to-linear complexity. IRBSOG is the fastest across all tested system sizes: it achieves approximately an order-of-magnitude speedup over PPPM and is roughly twice as fast as the RBSOG method.

\begin{figure}[!ht]
	\centering	
 	\includegraphics[width=0.8\linewidth]{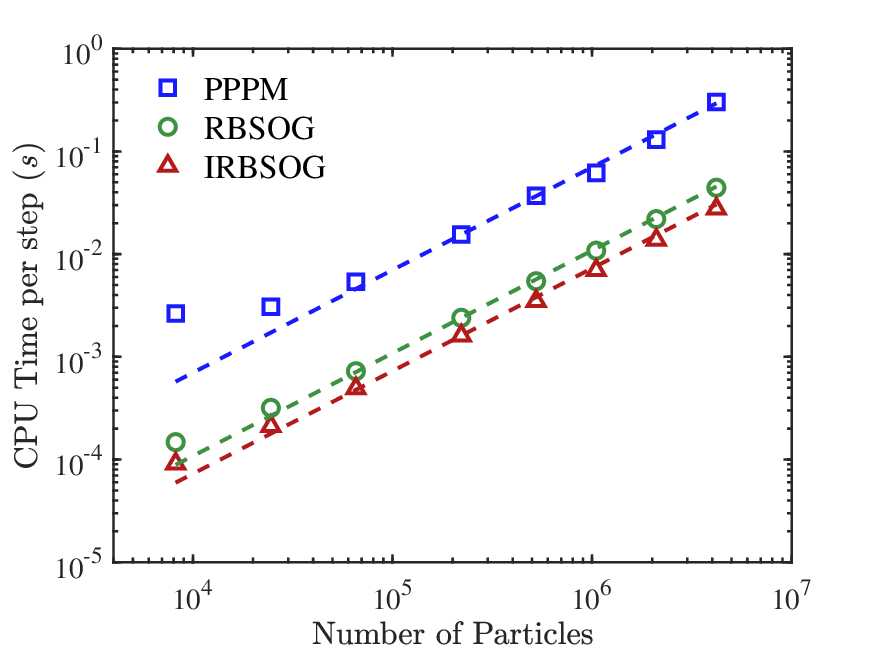}
	\captionsetup{justification=raggedright, singlelinecheck=false}
	\caption{The CPU time per step of the PPPM, RBSOG with $P_F=100$ and IRBSOG with $P_R=30$, $P_F=100$ against increasing number of particles using 1000 CPU Cores. The dashed lines show the linear fitting of data.}
	\label{linear}
\end{figure}

Next, we compare IRBSOG with PPPM on strong and weak parallel scalability, where results are shown in Figs.~\ref{cputime_strong}a-d. With one node, the per-step CPU time of IRBSOG is already much lower than that of PPPM. As node count increases, IRBSOG shows even much better strong-scaling: at $16$ nodes IRBSOG still attains about $90\%$ strong scaling, while PPPM drops to $\sim40\%$. IRBSOG also maintains near-ideal weak scaling: its efficiency stays close to $100\%$ across the tested node counts, whereas PPPM declines quickly. At $16$ nodes, IRBSOG is at $\sim98\%$ weak scaling and PPPM is below $50\%$. To test performance on large-scale systems, we measure CPU time for a NaCl system with $N=1.024\times10^6$. Using $\sim 1000$ cores, IRBSOG achieves an order-of-magnitude speedup over PPPM for the total CPU time, demonstrating its attractive efficiency and scalability.

In Table~\ref{tabrbl}, we break down the CPU time of the short- and long-range parts for IRBSOG and PPPM and report the corresponding speedup factors. With 100 CPU cores, IRBSOG attains speedups of 1.82 (short-range) and 16.34 (long-range), and the wall-clock times of these two parts are comparable within IRBSOG. When the core count increases to 1000, the short-range speedup remains similar, whereas the long-range speedup rises to 22.72, indicating lower communication overhead and better strong scaling for IRBSOG. For the short-range BMH term, we use the same cutoff, $r_c=10~\text{\AA}$, for both methods. Although increasing $r_c$ could in principle shift more work to the real-space part and better balance PPPM, it also enlarges the neighbor list and memory storage, which limits the accessible system size. Even with such tuning, our tests suggest that IRBSOG would still deliver about a 3-4$\times$ speedup while using less memory usage.

\begin{figure}[!ht]
	\centering	
	\includegraphics[width=0.8\linewidth]{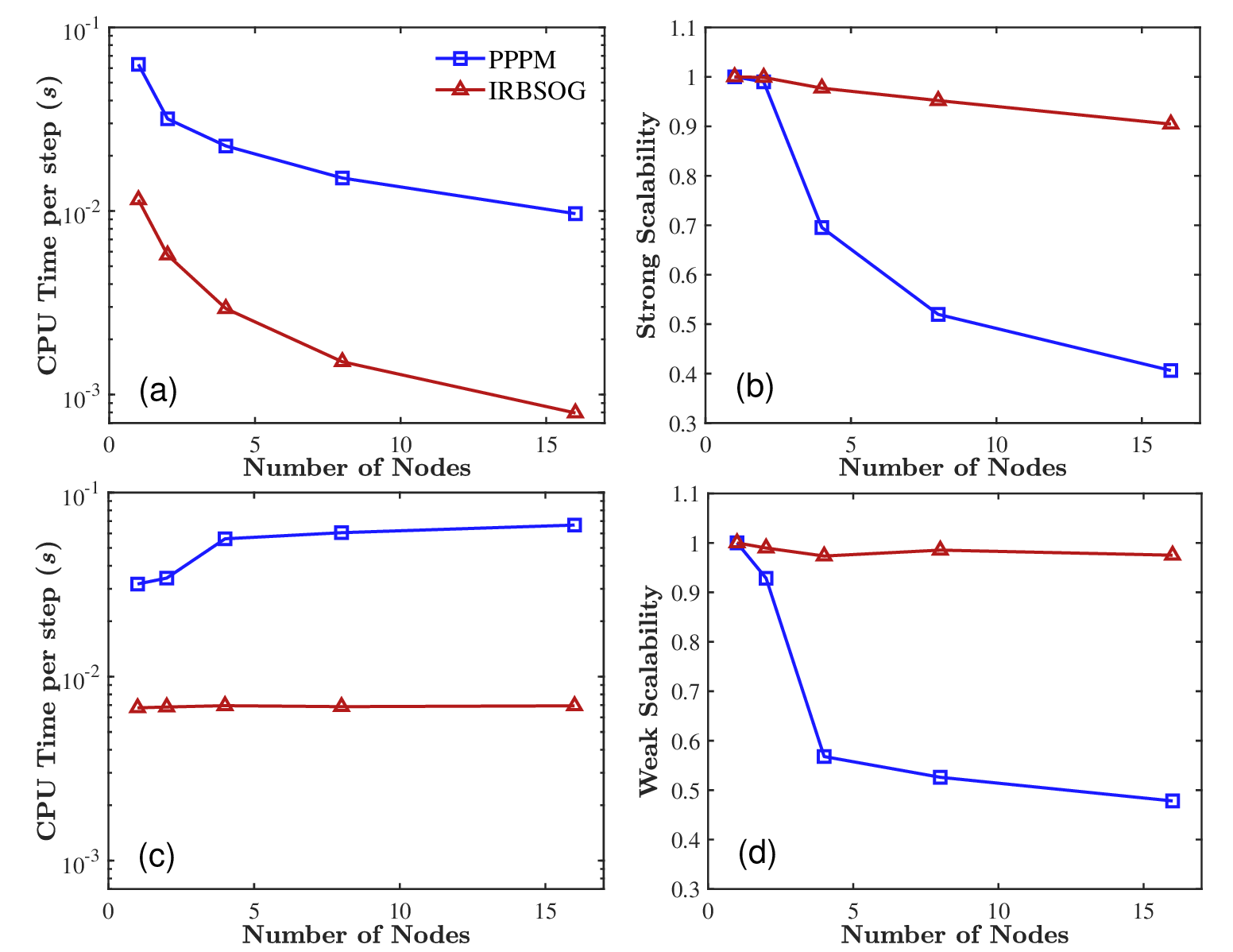}
	\captionsetup{justification=raggedright, singlelinecheck=false}
	\caption{CPU time and strong/weak scalability for the IRBSOG and PPPM, using up to \(16\) nodes with 128 CPU core per node. (a-b) present the strong scalability results with a fixed total particle number of $221184$. (c-d) present the weak scalability results with an average of $1024$ particles per core. }
	\label{cputime_strong}
\end{figure}

\begin{table}[t]
    \caption{Average CPU time (ms) per step for Pair, Kspace and Total as a function of different number of CPU cores produced by the IRBSOG ($P_R=30$ and $P_F=100$) with $r_{\text{core}}=5~\mathring{A}$ and $r_{\text{shell}}=10~\mathring{A}$, the PPPM with $r_c=10~\mathring{A}$ and the optimal PPPM with $r_c=12~\mathring{A}$ for the NaCl system with $N=1.024\times10^6$. The memory usage (Mbytes) using 100 CPU cores is also listed.}
    \setlength{\tabcolsep}{5pt}
    \renewcommand{\arraystretch}{1.2}
    \centering
        \resizebox{\textwidth}{!}{
    \begin{tabular}{|c|c|c|c|c|c|c|c|c|c|c|c|}
        \hline
        \multicolumn{2}{|c|}{\multirow{3}{*}{}}  & \multicolumn{9}{c|}{Number of CPU Cores}  & \multirow{3}{*}{\makecell{Memory}} \\
        \cline{3-11} 
        \multicolumn{2}{|c|}{} &\multicolumn{3}{c|}{100} &\multicolumn{3}{c|}{500} &\multicolumn{3}{c|}{1000} & \multicolumn{1}{c|}{\multirow{2}{*}{}} \\
        \cline{3-11} 
        \multicolumn{2}{|c|}{}  &Pair &Kspace &Total 
        &Pair &Kspace &Total
        &Pair &Kspace &Total & \multicolumn{1}{c|}{} \\
        \hline
        \multicolumn{2}{|c|}{RBSOG}    
        &58.42 &26.01  &84.43    
        &12.91  &5.25  &18.16
        &6.37  &2.67   &9.04 &20.32
        \\ 
        \hline
        \multicolumn{2}{|c|}{PPPM} 
        &79.94 &312.02  &391.96
        &15.90 &89.75  &105.65
        &8.12 &58.62  &66.74  &49.64
        \\
        \hline
        \multicolumn{2}{|c|}{optimal PPPM} 
        &132.18 &187.11  &319.29 
        &26.43 &48.35  &74.78
        &13.52 &33.41  &46.93  &43.86
        \\
        \hline
    \end{tabular}}
    \label{tabrbl}
\end{table}

\subsection{Tests on molten alkali binary halide systems}
Molten alkali binary halides, although composed of only three ionic species, show complex thermodynamic and transport behavior beyond that of simple molten salts or dilute electrolytes. Key properties--including melting point, viscosity, and electrical conductivity--can vary nonmonotonically with composition, reflecting strong ion-ion correlations and composition-dependent structural reorganization in the melt. Accurately capturing these effects is therefore a stringent test for interatomic models used at high temperature.

To assess the performance of the IRBSOG method, we conduct MD simulations of the LiCl-NaCl mixture over a wide range of compositions. The potential parameters are summarized in Table~\ref{tab3}. For IRBSOG, we set $r_{\mathrm{core}}=5~\text{\AA}$, $r_{\mathrm{shell}}=10~\text{\AA}$, and use batch sizes $P_R=30$ and $P_F=100$. Figure~\ref{rdf3} shows the RDFs for six representative compositions. IRBSOG reproduces composition-dependent shifts in peak positions and changes in coordination features across the first and second shells, as well as medium-range correlations. These results indicate that IRBSOG maintains accuracy and transferability for chemically mixed molten salts in which cross-species interactions play a central role.

\begin{table*}[!ht]
	\caption{BMH potential parameters for molten LiCl-NaCl.}
		\setlength{\tabcolsep}{8pt}
	\renewcommand{\arraystretch}{1.3}
	\centering
	\begin{tabular}{ccccccc}
		\hline
		\  
		&Li-Li &Li-Na &Na-Na &Li-Cl &Na-Cl &Cl-Cl \\ \hline
		A$/$eV    &0.4225 &0.334 &0.2640625  &0.29046875 &0.21125  &0.1584375 \\ \hline
		$\sigma$$/$$\mathring{A}$    &1.632 &1.986 &2.340  &2.401 &2.755 &3.170 \\ \hline
		C$/$eV$\cdot\mathring{A}^6$ 	  &0.045625 &0.218875 &1.05  &1.25 &7.00  &71.328 \\  \hline
		D$/$eV$\cdot\mathring{A}^8$  	& 0.01875 &0.0968 &0.50  &1.50 & 8.6875 & 143.281\\  \hline
		$\rho$$/$$\mathring{A}$ 	  &0.3170 &0.3185  &0.3200  &0.3215 &0.3230  &0.3261 \\  \hline
	\end{tabular}
	\label{tab3}
\end{table*}

\begin{figure}[!ht]
	\centering	
        \includegraphics[width=0.8\linewidth]{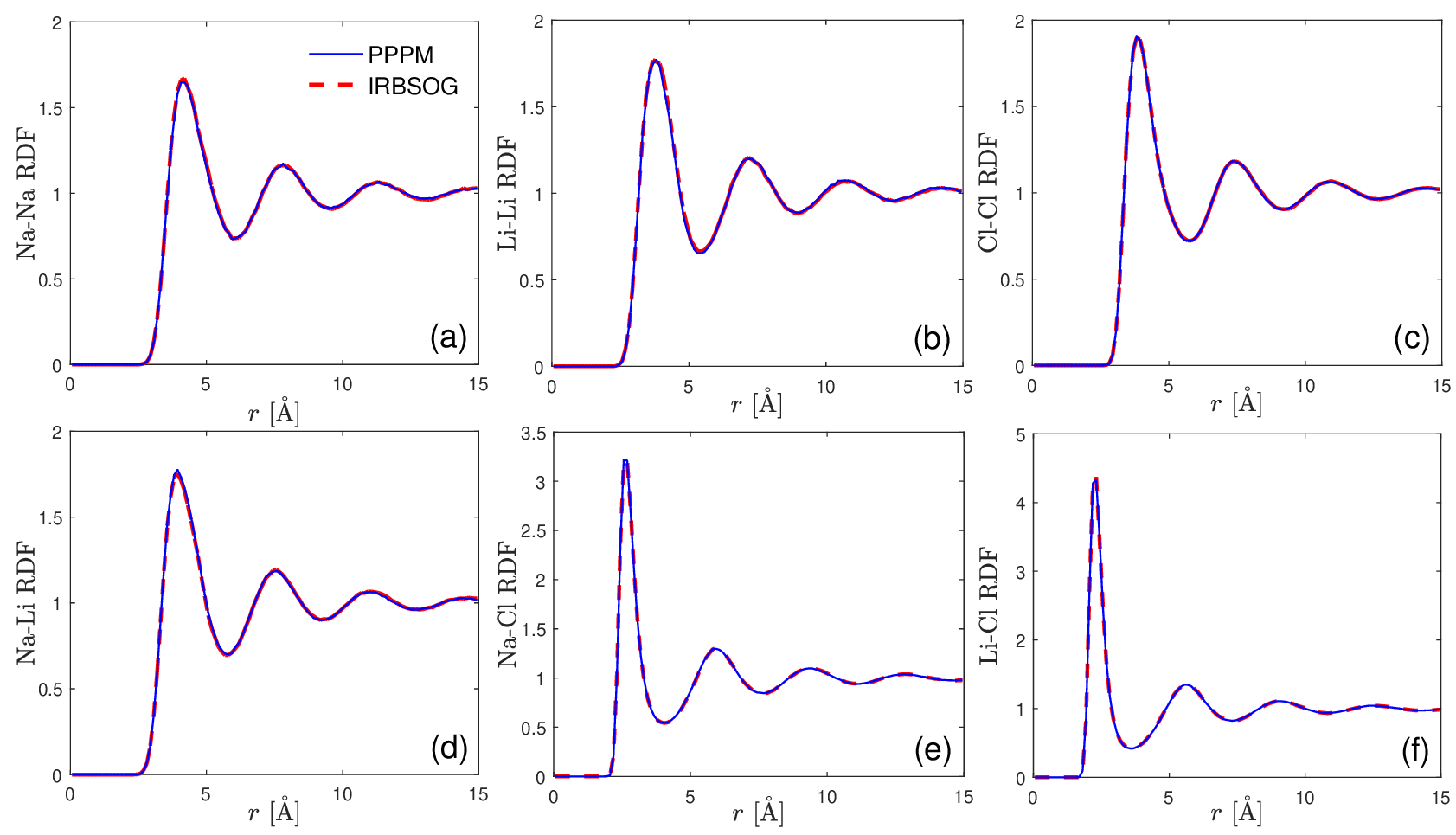}
    \captionsetup{justification=raggedright, singlelinecheck=false}
	\caption{The RDFs of Cl-Cl(a), Li-Li(b), Na-Na(c), Cl-Li(d) and Na-Cl(e), and the visual model of LiCl-NaCl(f) at 1100 K where the purple, red and green balls represent Na, Li and Cl, respectively. We use the IRBSOG with $P_R=30$ and $P_F=100$ compared to the PPPM. }
	\label{rdf3}
\end{figure}

\section{Conclusion}\label{sec:IV}
In this work, we develop an IRBSOG method for MD simulations of BMH systems. The IRBSOG framework adopts a decomposition strategy based on an SOG decomposition of the Coulomb kernel. It is a stochastic scheme that couples random mini-batch sampling for short-range neighbor interactions with random-batch importance sampling for the long-range component, leading to a total 
computational complexity $O(N)$. Compared with FFT-accelerated mesh Ewald methods and other random-batch approaches, IRBSOG exhibits improved computational efficiency and reduced memory consumption. Numerical simulations of molten NaCl and alkali binary halide systems demonstrate its accuracy and scalability. In addition, the method can be naturally extended to quasi-two-dimensional systems with planar interfaces and dielectric mismatch~\cite{gan2025random}, which will be investigated in future work.

\section*{Acknowledgments}
The work of C.C., J.L. and Z.X. was supported by the National Natural Science Foundation of China (grants No. 12426304, 12325113 and 12401570) and the Science and Technology Commission of Shanghai Municipality (grant No. 23JC1402300).  The authors would like to thank the support from the SJTU Kunpeng \& Ascend Center of Excellence.

\bibliographystyle{ieeetr}  
\bibliography{ref}
\end{document}